\documentclass[11pt]{article}
\usepackage[utf8]{inputenc}
\usepackage[english]{babel}
\usepackage{amsmath,amssymb,mathrsfs,amsthm}
\usepackage[filecolor=green]{hyperref}
\usepackage{color}
\usepackage{tikz}
\usepackage{tikz-cd}
\usepackage{enumitem}
\usepackage{vmargin}
\setmarginsrb{2cm}{2cm}{2cm}{1.5cm}{0cm}{0cm}{0cm}{1.5cm}
\newtheorem{theorem}{Theorem}[section]
\newtheorem{lemma}[theorem]{Lemma}
\newtheorem{proposition}[theorem]{Proposition}

\theoremstyle{definition}
\newtheorem{definition}[theorem]{Definition}
\newtheorem{remark}[theorem]{Remark}
\newtheorem{example}[theorem]{Example}

\newcommand\eg{\emph{e.g.}}
\newcommand\ie{\emph{i.e.}}
\newcommand\wrt{w.r.t.}
\newcommand\G{Gröbner}

\newcommand\diff[1]{\partial_{#1}}
\newcommand\D{\mathcal{D}}
\newcommand\basis{\mathscr{B}}

\newcommand\ordS{<_S}
\newcommand\ordR{<_R}
\newcommand\Rrho{R_{\prec\rho}}
\DeclareMathOperator{\supp}{supp}
\DeclareMathOperator{\lm}{lm}
\DeclareMathOperator{\lc}{lc}
\DeclareMathOperator{\id}{id}

\newcommand\K{\mathbb{K}}
\newcommand\Q{\mathbb{Q}}
\newcommand\R{\mathbb{R}} 
\newcommand\N{\mathbb{N}}
\newcommand\QX{\mathbb{Q}[x_1,\cdots,x_n]}
\newcommand\QXX{\mathbb{Q}(x_1,\cdots,x_n)}
\newcommand\Weyl[1]{B_{#1}(\Q)}
\newcommand\monBasis{\Mon(\Delta)}
\newcommand\Span[1]{\langle #1\rangle}
\DeclareMathOperator{\Mon}{Mon}
\DeclareMathOperator{\im}{im}

\newcommand\rewR{\to_R}

\newcommand\parS{\twoheadrightarrow_S}
\newcommand\parTheta[1]{\twoheadrightarrow_{\Theta,#1}}

\newcommand\transRew{\overset{*}{\to}}
\newcommand\transR{\overset{*}{\to}_R}
\newcommand\transS{\overset{*}{\twoheadrightarrow}_S}

\newcommand\rewEquiv{\overset{*}{\leftrightarrow}}
\newcommand\equivR{\overset{*}{\leftrightarrow}_R}

\DeclareMathOperator{\SNF}{{\it S}-NF}

\newcommand\rewTheta{\to_\Theta}
\newcommand\transTheta{\overset{*}{\to}_\Theta}
\newcommand\divInv[1]{\mid_{#1}}
\newcommand\RTheta{R_{\Theta}}
\newcommand\SThetaL{S_{\Theta,L}}
\newcommand{\SThetaNF}[1]{{\it S}_{\Theta,#1}\operatorname{-NF}}
\DeclareMathOperator{\SThetaLNF}{{\it S}_{\Theta,L}-NF}

\begin{document}

\title{Strategies for linear rewriting systems:\\[0.2cm]
 link with parallel rewriting and involutive divisions\vspace{0.5cm}}
\author{Cyrille Chenavier\footnote{
    Johannes Kepler University, Institute for Algebra,
    cyrille.chenavier@jku.at.
  }\and Maxime Lucas\footnote{Inria Rennes - Bretagne Atlantique,
    Gallinette team, maxime.lucas@inria.fr.}
  }
\date{}

\maketitle
      
\begin{abstract}
  We study rewriting systems whose underlying set of terms is equipped
  with a vector space structure over a given field. We introduce parallel
  rewriting relations, which are rewriting relations compatible with the
  vector space structure, as well as rewriting strategies, which consist
  in choosing one rewriting step for each reducible basis element of the
  vector space. Using these notions, we introduce the $S$-confluence
  property and show that it implies confluence. We deduce a proof of the
  diamond lemma, based on strategies. We illustrate our general framework
  with rewriting systems over rational Weyl algebras, that are vector
  spaces over a field of rational functions. In particular, we show that
  involutive divisions induce rewriting strategies over rational Weyl
  algebras, and using the $S$-confluence property, we show that
  involutive sets induce confluent rewriting systems over rational Weyl
  algebras.
\end{abstract}
\noindent
\begin{small}\textbf{Keywords:} confluence, parallel rewriting, rewriting
  strategies, involutive divisions.\\[0.2cm]
  \textbf{M.S.C 2010 - Primary:} 13N10, 68Q42. \textbf{Secondary:} 12H05,
  35A25.
\end{small}

\tableofcontents

\section{Introduction}

Rewriting systems are computational models given by a set of syntactic
expressions and transformation rules used to simplify expressions into
equivalent ones. Since rewriting theory is applicable to different
problems of mathematics and computer science, it was developed for many
syntaxes of terms, \eg, strings, ($\Sigma-$, higher-order, infinitary)
terms, graphs, (commutative, noncommutative, vectors of) polynomials,
(linear combinations of) trees, (higher-dimensional) cells. Abstract
rewriting theory unifies these contexts and provides universal
formulations of rewriting properties, such as termination, normalisation
and (local) confluence. Newman's lemma is one of the most famous results 
of abstract rewriting and asserts that under termination hypothesis, 
local confluence implies confluence.
\medskip

In the context of rewriting over algebraic structures, the Newman's lemma
is used in conjunction with the critical pairs lemma to algorithmically
prove confluence. This is something fundamental since confluent rewriting
systems provide methods for solving decision problems, computing (linear,
homotopy) bases, Hilbert series, or free
resolutions~\cite{MR846601, GuiraudHoffbeckMalbos19, MR2964639,
  MR1072284, MR1360005}. From these methods, one get constructive proofs
of theoretical results, such as  embedding, coherence or homological
theorems~\cite{MR506890, MR0506423, MR3347996, MR3742562, MR265437,
  MR920522}, but also applications to problems coming from topics
modelled by algebra, such as cryptography, analysis of (ordinary
differential, partial derivative, time-delay) equations or control
theory. For instance, many informations of functional equations may be
read over free resolutions: integrability conditions, parametrization of
solutions, existence of autonomous curves~\cite{MR2233761, MR1308976}.
\medskip

When one considers algebraic structures with underlying vector space
operations, the conjunction of Newman's lemma and the critical pairs
lemma is traditionally known under the name of the diamond lemma. In
practice, this lemma is used to test if a generating set of a polynomial
ideal is a \G\ basis, since confluent linear rewriting systems are
usually induced by \G\ bases or one of their numerous adaptations to
different classes of algebras or operads~\cite{MR506890, MR2202562,
  MR2667136, MR1044911, MR1299371}. As an illustration of theses classes,
let us mention polynomial Weyl algebras that are models of differential
operators with polynomial coefficients. These algebras are composed of
polynomials over two sets of $n$ variables, the state variables
$x_1,\ldots,x_n$ and the vector field variables
$\partial_1,\ldots,\partial_n$, and submitted to the commutation rules
\[ \forall 1\leq i\neq j\leq n: \qquad x_ix_j=x_jx_i,
\quad \partial_i\partial_j=\partial_j\partial_i, \quad \partial_ix_j=
x_j\partial_i, \quad\partial_ix_i=x_i\partial_i+1. \]
These relations represent classical rules from differential calculus: the
second one means that second order derivatives of smooth functions
commute, the third one means that $x_j$ is constant for differentiation
with respect to $x_i$ and the last one represents Leibniz's rule for
differentiation with respect to $x_i$, that is,
$\partial_i(x_if)=x_i\partial_i(f)+f$, for any smooth function $f$.
Rewriting over vector spaces requires to introduce a notion of
well-formed rewriting step, also called {\em positive rewriting step}
in~\cite{GuiraudHoffbeckMalbos19}, that is, each step reduces only one
basis element together with its coefficient in a given vector. Typically,
the vector is a polynomial and the reduced basis element is a monomial
that appears in this polynomial with a nonzero coefficient. Doing so
avoids pathological situations, \eg, $v_1\to v_2$ implies that the
rewriting step $v_2=v_2-v_1+v_1\to v_2-v_2+v_1=v_1$ is not well-formed.
On the other hand, well-formed rewriting steps are not compatible with
vector space operations since as soon as two different basis elements are
rewritten in the well-formed rewriting steps $u_1\to u_2$ and
$v_1\to v_2$, then $u_1+v_1\to u_2+v_2$ is not well-formed.
\medskip

The notion of well-formed rewriting step is specific to rewriting systems over
vector spaces and, as mentioned above, it is not compatible with the underlying
algebraic operations. This lack of compatibility makes the theory of linear
rewriting rather painful, for instance the proof of the critical pair lemma is
more involved than for string rewriting (see~\cite[Theorem
4.2.1]{GuiraudHoffbeckMalbos19}). In our view, these observations call for the
development of a theory of linear rewriting free from well-formed rewriting
steps. In the long run, we hope that this will contribute to bridge the gap
between abstract and linear rewriting, for instance to find a common proof of
Newman's lemma and of the diamond lemma.

\begin{center}
  {\large\bf Our results}
\end{center}

In the present paper, we introduce an alternative approach to rewriting theory
over vector spaces, which does not use the notion of well-formed rewriting
step. Instead, our rewriting steps only depend on vector spaces operations. A
nice feature is that we allow the reduction of many basis elements at once,
while still avoiding pathological situations. From this approach, we get a
method for constructing bases of quotient vector spaces. We also relate our
approach to the classical one based on well-formed rewriting steps. In
particular, we get a confluence criterion which extends the Diamond
Lemma. Moreover, our framework is valid in every vector space and can be applied
to the case where the coefficients do not commute with basis elements. We
illustrate this last point with rewriting systems over rational Weyl algebras,
and from this, we show that so-called {\em involutive bases} \cite{MR1627129}
induce confluent rewriting systems.  \medskip

Consider a vector space $V$ with a basis $\basis$, and a set $R$ of
rewriting rules of the form $e\to v$, where $e\in\basis$ and $v\in V$.
Our approach to rewriting over vector spaces first consists in selecting,
for each $e\in\basis$, at most one rewriting rule with left-hand side
$e$. The set of such chosen rules is denoted by $S$. We consider the
parallel rewriting relation induced by $S$ whose steps are of the form
$v_1\parS v_2$ and are defined in such  a way that all selected left-hand
sides of elements of $S$ occurring in $v_1$ are replaced by the
corresponding right-hand sides to get $v_2$. As mentioned above, the
definition of parallel rewriting is purely internal to the category of
vector spaces and does not require any notion of well-formed rewriting
step. We associate to $S$ a preorder, denoted $\ordS$. In the case where $\ordS$ is
well-founded, we say that $S$ is a {\em strategy}. Following the ideas
of~\cite{GuiraudHoffbeckMalbos19}, and contrary to rewriting with \G\
bases, rewriting with strategies does not require any monomial order but
uses the order induced by the rewriting process itself. Given a strategy
$S$, we introduce the $S$-confluence property that asserts that for every
rule $e\rewR v$ in $S$ and all other rule $e\rewR u$ of $R$ that is not
in $S$, then $u$ and $v$ are joinable using~$\parS$. In
Theorem~\ref{thm:SNF_decompo}, we show under $S$-confluence hypothesis,
a basis of $V$ modulo the subspace generated by elements $e-v$, where
$e\rewR v\in R$, admits as a basis elements of $\basis$ that are not
left-hand sides of $S$. We also relate this approach based on strategies
to the classical one based on well-formed rewriting steps. In particular,
we show in Theorem~\ref{thm-S-conf} that $S$-confluence implies
confluence of the rewriting relation defined using only well-formed
rewriting steps. As a consequence, we give in
Theorem~\ref{thm:diamond_lemma} a new proof of the diamond lemma based on
strategies.
\medskip

Since our approach works for arbitrary vector spaces, it may be declined
in different classes of algebras over fields, including rational Weyl
algebras. They are composed of differential operators with coefficients in
the field of rational functions. Unlike polynomial Weyl algebras, which
are vector spaces over the field of constants and modules over the ring
of polynomials, rational Weyl algebras are vector spaces over the field
of rational functions. Notice that because of Leibniz's rule,
rational functions do not commute with operators. Rewriting-like methods
in this context yield applications to formal analysis of linear systems
of ordinary differential or partial derivative equations as mentioned
above. In particular, involutive divisions, such as {\em Janet, Thomas}
and {\em Pommaret} divisions, provide deterministic techniques to rewrite
differential operators. By determinism, we mean that each differential
operator admits at most one involutive divisor, that is, it may be
rewritten into at most one other differential operator. This determinism
is strongly related to our notion of strategy and parallel reductions. In
particular, we show in  Theorem~\ref{thm:involutive_conf} that involutive
bases induce $S$-confluent rewriting systems. Finally, we show how most
of the axioms of involutive divisions may be formalised in purely
rewriting language using rewriting strategies.
\medskip

\paragraph{Organisation.}

In Section~\ref{sec:rewriting_strategies_over_vector_space}, we present
our general framework and results for rewriting systems over vector
spaces. In Section~\ref{sec:confluence_relative_to_a_strategy}, we 
introduce rewriting strategies and the $S$-confluence property. We show
that $S$-confluence guarantees that bases of quotient vector spaces are
induced by normal forms basis elements for the strategy. In
Section~\ref{sec:strategies_for_traditional_rewriting_relations}, we
relate rewriting strategies to the classical approach to rewriting theory
on vector spaces. In particular, we show that $S$-confluence implies
confluence of the classical rewriting relation, and obtain a new proof of
the Diamond Lemma. In
Section~\ref{sec:rewriting_strategies_over_rational_Weyl_algebras}, we
illustrate our general framework by rewriting systems over rational 
Weyl algebras. In Section~\ref{sec:rewriting_systems_over_Weyl_algebras},
we introduce well-formed rewriting systems over rational Weyl algebras. 
In Section~\ref{sec:involutive_divisions_and_strategies}, we recall the
definition of an involutive division and of involutive sets of operators
and show that involutive divisions define rewriting strategies, and that
involutive sets $S$-confluent rewriting relations.
\medskip

\paragraph{Terminology and conventions.}

Throughout the paper, we use the standard terminology and conventions of
rewriting theory~\cite{MR1629216}. An {\em abstract rewriting system} is
a pair $(A,\to)$, where $A$ is a set and $\to$ is binary relation on $A$,
called {\em rewriting relation}. An element~$(a,b)\in\to$ is written
$a\to b$ and is called a {\em rewriting step}. A {\em normal form} for
$\to$ is an element $a\in A$ such that there is no $b\in A$ such that
$a\to b$. We denote by $\transRew$  (respectively,~$\rewEquiv$) the
closure of $\to$ under transitivity and reflexivity (respectively, and
symmetry). The equivalence class of $a\in A$ modulo the equivalence
relation $\rewEquiv$ is written $[a]_{\rewEquiv}$ and the set of all
equivalence classes is written $A/\rewEquiv$. When $a\transRew b$, that
is, there exists a (possibly empty) finite sequence of rewriting steps
from $a$ to  $b$, we say that $a$ {\em rewrites} into $b$. The rewriting
relation $\to$ is said to be {\em confluent} if whenever $a\transRew b$
and $a\transRew c$, then $b$ and $c$ are {\em joinable}, that is, there
exists $d$ such that $b\transRew c$ and $c\transRew d$. The confluence
property is equivalent to the {\em Church-Rosser property} that asserts
that whenever $a\rewEquiv b$, then~$a$ and $b$ are joinable.
\medskip

We also recall the notion of support of a vector in a given basis of a
vector space. Given a vector space $V$ over the ground field $\K$ and a
$\basis$ of $V$, every vector $v$ admits a unique finite decomposition
with respect to the basis $\basis$ and coefficients in the ground field:

\begin{equation}\label{equ:vector_decompo}
  v=\sum\lambda_ie_i,\quad\lambda_i\neq 0.
\end{equation}
The set of basis elements which appear in the decomposition
\eqref{equ:vector_decompo} is called the {\it support} of $v$ and is
written $\supp(v)$.
\medskip

\paragraph{Acknowledgement.} This work was supported by the 
Austrian  Science  Fund  (FWF):  P 32301. 

\section{Rewriting strategies over vector spaces}
\label{sec:rewriting_strategies_over_vector_space}

In this section, we introduce rewriting strategies over vector spaces and
relate them to the traditional approach to rewriting theory over vector
spaces. 
\medskip

Throughout the section, we fix a ground field $\K$, a $\K$-vector space $V$, and
a basis $\basis$ of $V$. We respectively call elements of $V$ and $\basis$
vectors and basis elements. We fix a subset $R$ of $\basis\times V$. The set $R$
represents a set of rewriting rules over~$V$, and for this reason, we denote its
elements by $e\rewR v$ instead of $(e,v)$ and call them \emph{rewriting
  rules}. We denote by $\Span{R}$ the subspace of $V$ spanned by elements $e-v$,
where $e\rewR v$ is a rewriting rule. We write $V/\Span{R}$ for the quotient of
$V$ by $\Span{R}$.

\subsection{Confluence relative to a strategy}
\label{sec:confluence_relative_to_a_strategy}

In this section, we define rewriting strategies as well as the $S$-confluence
property, and show that under $S$-confluence hypothesis, a basis of $V/\Span{R}$
may be described in terms of normal-form-like basis vectors.  \medskip

\medskip

We fix a subset~$S$ of~$R$ with pairwise distinct left-hand sides, that is,
given two rules $e\rewR v$ and $e'\rewR v'$ in $S$ such that $e=e'$, we have
$e\rewR v=e'\rewR v'$. Let us consider the endomorphism $r_S:V\to V$ defined by
$r_S(e)=v$ if there exists a rule $e\rewR v$ in $S$ with left-hand side $e$, and
$r_S(e)=e$ if there is no rule in $S$ with left-hand side $e$. We consider the
rewriting relation $\parS$ on $V$ defined by $u\parS u'$ whenever
$u'=r_S(u)$. In the sequel, we refer this relation as being the parallel
rewriting relation induced by $S$; the terminology parallel makes explicit that
all basis vectors in $\supp(u)$ are reduced at once with rules in $S$ and
justifies the double head arrows in $\parS$. Notice that the parallel rewriting
relation is stable under vector spaces operations, that is, for every rewriting
steps $u\parS u'$, $v\parS v'$, and for every scalar $\lambda\in\K$, there exists a
rewriting step $u+\lambda v\parS u'+\lambda v'$. Moreover, also notice that
$\parS$ is deterministic in the sense that for every vector $u$, there is
exactly one~$u'$ such that $u\parS u'$. Let us define the preorder of $S$ as
being the transitive closure of the binary relation defined by $e'\ordS e$ if
there exists an $n > 0$ such that $e'\in\supp(r_S^n(e))$ and $r_S^n(e) \neq e$, where
$r_S^n$ is the $n$-th composition of $r_S$. The preorder of $S$ is also
written $\ordS$.
\smallskip

\begin{definition}\label{def:strategies}
  A {\em prestrategy} for $R$ is a subset $S$ of $R$ with pairwise
  distinct left-hand sides. A {\em strategy} for $R$ is a prestrategy
  such that the preorder $\ordS$ of that prestrategy $S$ is well-founded. 
\end{definition}
\smallskip

Notice that if $S$ is a strategy for $R$, then $\ordS$ is a well-founded
order. Moreover, every vector $u$ admits an {\em S-normal form}, that is
a vector $u'$ such that $u\transS u'$ and $r_S$ acts trivially on $u'$.
Indeed, since $\ordS$ is well-founded, the sequence $(r_S^n(u))_n$ is
stationary and $u'$ is the limit of this sequence. Moreover, since
$\parS$ is deterministic, this $S$-normal form is unique and we denote it
by~$\SNF(u)$. 
\smallskip

\begin{example}\label{ex:strategies_step_1}
  Let us illustrate (pre)strategies with the following
  (counter-)examples.
  \begin{enumerate}
  \item\label{it:ex_strat_1} Assume that $V$ is $4$-dimensional with
    basis $\basis=\{e_1,e_2,e_3,e_4\}$ and let us consider the following
    set of rewriting rules:
    \[R=\{e_1\rewR e_2,\quad e_2\rewR e_3+e_4,\quad e_3\rewR e_2-e_4\}.\]
    Considering the prestrategy
    \[S=\{e_1\rewR e_2,\quad e_2\rewR e_3+e_4\}\subset R,\]
    the preorder $\ordS$ is well-founded since $e_1>_Se_2>_Se_3,e_4$, and
    $e_3,e_4$ are minimal. Hence,~$S$ is a strategy for $R$. The unique
    $S$-normal form of $u=\lambda_1e_1+\cdots+\lambda_4e_4$ is computed
    as follows:
    \[u\parS\lambda_1e_2+(\lambda_2+\lambda_3)e_3+
    (\lambda_2+\lambda_4)e_4\parS(\lambda_1+\lambda_2+\lambda_3)e_3+
    (\lambda_1+\lambda_2+\lambda_4)e_4,\]
    which yields:
    \[\SNF(u)=(\lambda_1+\lambda_2+\lambda_3)e_3+
    (\lambda_1+\lambda_2+\lambda_4)e_4.\]
  \item\label{it:c-ex_strat_1} Assume that $V$ is $3$-dimensional with
    basis $\basis=\{e_1,e_2,e_3\}$, let us consider the set of rewriting
    rules:
    \[R=\{e_1\rewR e_2+e_3,\quad e_2\rewR e_1,\quad e_3\rewR -e_1\},\]
    and $S=R$. Then, $S$ is not a strategy since the preorder $\ordS$ is
    cyclic: from the rewriting sequence $e_2\parS e_1\parS e_2+e_3$, we
    get $e_2>_Se_3$, so that $e_1>_S e_2>_S e_3>_S e_1>_S\cdots$. Notice
    however that each vector $u=\lambda_1e_1+\lambda_2e_2+\lambda_3e_3$
    admits a unique $S$-normal form which is $0$:
    \[u=\lambda_1e_1+\lambda_2e_2+\lambda_3e_3\parS(\lambda_2-\lambda_3)
    e_1+\lambda_1e_2+\lambda_1e_3\parS(\lambda_2-\lambda_3)e_2+(\lambda_2-
    \lambda_3)e_3\parS 0.
    \]
  \item\label{it:case_N} Let $V$ be the vector space with basis
    $\basis= \mathbb N$ and consider the set of rewriting rules:
    \[R=\{n\rewR n+1:\quad n\in\N\}.\]
    Then, a prestrategy $S$ corresponds to a subset $E$ of $\mathbb N$.
    Moreover, such a prestrategy is a strategy if and only if for all
    $n\in E$, there exists $k\in \mathbb N$ such that $n + k \notin E$. 
  \end{enumerate}
\end{example}
\smallskip

From now on, we fix a strategy $S$ for $R$ and let us denote by $\SNF$
the map from $V$ to itself that maps any vector $u$ to $\SNF(u)$.
\smallskip

\begin{proposition}\label{prop:SNF_projector}
  The map $\SNF$ is a linear projector.
\end{proposition}

\begin{proof}
  For every vector $u\in V$, we have $r_S(\SNF(u))=\SNF(u)$, which proves 
  $\SNF\circ\SNF=\SNF$. Moreover, given another vector $v\in V$, let
  $k\in\N$ be an integer such that $\SNF(u)=r_S^k(u)$ and
  $\SNF(v)=r_S^k(v)$. For every scalar $\lambda\in\K$, we have 
  \[r^k_S(u+\lambda v)=r^k_S(u)+\lambda r^k_S(v)=\SNF(u)+\lambda\SNF(v).
  \]
  Hence, $\SNF(u+\lambda v)=r^k_S(u+\lambda v)$ is equal to
  $\SNF(u)+\lambda\SNF(v)$, which proves that $\SNF$ is linear.
\end{proof}

\smallskip

Now, we introduce the $S$-confluence property. 
\smallskip

\begin{definition}\label{def:standardisation_property}
  Given a strategy $S$ for $R$, we say that $R$ is \emph{S-confluent} if
  for every rewriting rule $e\rewR v$ in $R$, we have $\SNF(e-v)=0$.
\end{definition}
\smallskip

In the following theorem, we show that under $S$-confluence hypothesis,
the operator $\SNF$ induces the natural projection from $V$ to
$V/\Span{R}$.
\medskip

\begin{theorem}\label{thm:SNF_decompo}
  Let $R$ be a set of rewriting rules and let $S$ be a strategy for $R$.  If $R$
  is $S$-confluent, then we have the following equality and isomorphism of
  vector spaces:
  \begin{equation}\label{equ:iso_SNF}
    V/\Span{R} \quad = \quad V/\Span{S} \quad \simeq \quad\im(\SNF).
  \end{equation}
  In particular, $\{e+\Span{R}:\ \SNF(e)=e\}$ is a basis of
  $V/\Span{R}$.
\end{theorem}

\begin{proof} 
  The equality and isomorphism of~\eqref{equ:iso_SNF} are equivalent to 
  the following equalities:
  \begin{equation}\label{equ:R=S=ker_SNF}
    \ker(\SNF)=\Span{S}=\Span{R}.
  \end{equation}
  Let us show $\ker(\SNF)\subseteq\Span{S}$. From
  Proposition~\ref{prop:SNF_projector}, $\SNF$ is a projector, so
  that its kernel is equal to the image of the operator $\id_V-\SNF$.
  Moreover, by definition of the map $r_S$, for every vector $u\in V$, we
  have $u-r_S(u)\in\Span{S}$, which gives $u-\SNF(u)\in\Span{S}$ by
  induction on the smallest integer~$n$ such that $r_S^n(u)=\SNF(u)$.
  Hence, $\ker(\SNF)=\im(\id_V-\SNF)$ is included in $\Span{S}$.
  Moreover, the inclusion $\Span{S}\subseteq\Span{R}$ follows from the
  fact that $S$ is included in $R$. Let us show that
  $\Span{R}\subseteq\ker(\SNF)$. From the $S$-confluence hypothesis, for
  every rewriting rule $e\rewR v$, we have $\SNF(e-v)=0$. Since all
  elements $e-v$, for $e\rewR v$, generate $\Span{R}$, we  deduce that
  the latter is included in $\ker(\SNF)$, which concludes the proof 
  of~\eqref{equ:R=S=ker_SNF}, and shows~\eqref{equ:iso_SNF}. The second
  assertion of the theorem is a consequence of the fact that $\im(\SNF)$
  has a basis composed by basis elements that are fixed by $\SNF$.
\end{proof}
\smallskip

\begin{example}\label{ex:S-conf}
  Let us continue Point~\ref{it:ex_strat_1} of
  Example~\ref{ex:strategies_step_1}. The following identities hold:
  \[\begin{split}
  \SNF(e_1)=e_3+e_4=\SNF(e_2),&\quad\SNF(e_2)=e_3+e_4=\SNF(e_3+e_4)
  \\[0.3cm]
  \SNF(e_3)=& \ e_3=\SNF(e_2-e_4),
  \end{split}
  \smallskip\]
  so that $R$ is $S$-confluent. Notice that if we replace the rule
  $e_3\rewR e_2-e_4$ in $R$ by $e_3\rewR e_2$, we get $\SNF(e_3)=e_3$ and
  $\SNF(e_2)=e_3+e_4$, so $\rewR$ is not $S$-confluent anymore. 
\end{example}
\smallskip

We finish this section with a characterisation of the $S$-confluence
property in terms o decreasing. This characterisation is used to prove
the diamond lemma in terms of strategies, see
Theorem~\ref{thm:diamond_lemma}. Given an order $\prec$ on $R$ and a
rewriting rule $\rho\in R$, we denote by $\Rrho$ the set of rules
$\rho'\in R$ such that we have $\rho'\prec\rho$. In particular,
$\Span{\Rrho}$ is the subspace of $V$ generated by all the elements
$e-v$, where $e\rewR v\in R$ is strictly smaller than $\rho$ for $\prec$.
\smallskip

\begin{definition}
  Let $\prec$ be an order on $R$. We say that $R$ is {\em decreasing}
  \wrt\ $(S,\prec)$ if for any rewriting rule $\rho=e\rewR v \in R$, we 
  have $v-r_S(e)\in\Span{\Rrho}$.
\end{definition}
\smallskip

\begin{proposition}\label{prop:decreasing}
  The set $R$ is $S$-confluent if and only if there exists a well-founded
  order $\prec$ on $R$ such that $S$ is decreasing \wrt\ $(S,\prec)$.
\end{proposition}

\begin{proof}
  Assume that $R$ is $S$-confluent and let us define the order $\prec$ on
  $R$ by $\rho'\prec\rho$ if and only if~$\rho'\in R$ and $\rho\notin S$.
  This order is terminating since each strictly decreasing sequence has
  length $2$. Let $\rho=e\rewR v$ be a rewriting rule. We have
  $v-r_S(e)=(e-r_S(e))-(e-v)\in\Span{R}$, and from
  Theorem~\ref{thm:SNF_decompo}, if $R$ is $S$-confluent, then we have
  $\Span{R}=\Span{S}$, which yields $v-r_S(e)\in\Span{S}$. Moreover, 
  if~$\rho\in S$, then we have $v-r_S(e)=0$ and if not, we have
  $\Rrho=\Span{S}$. All together, we deduce that $R$ is decreasing \wrt\
  $(S,\prec)$.
  \smallskip
  
  Conversely, if $\prec$ is a well-founded order such that $R$ decreases
  \wrt\ $(S,\prec)$, then we get by induction that $\Span{R}$ is included
  in $\ker(\SNF)$. Indeed, if $\rho=e\rewR v\in R$ is a rewriting rule
  that is minimal for~$\prec$, that is $\Rrho=\emptyset$, then $v=r_S(e)$
  by decreasingness, so that $\SNF(e-v)=0$, and if $\rho$ is not minimal,
  by decreasingness we have $v-r_S(e)\in\Span{\Rrho}$, which is included
  in $\ker(\SNF)$ by induction hypothesis. Since the inclusion
  $\Span{R}\subseteq\ker(\SNF)$ is equivalent to the $S$-confluence
  property, the converse implication is shown.
\end{proof}
\smallskip

The proof of Proposition~\ref{prop:decreasing} indicates that for proving
$S$-confluence, we can choose the order that separates rewriting rules in
$S$ from the others. However, it may happen that other natural orders may
be used, as we will see in our proof of the diamond lemma.

\subsection{Strategies for traditional rewriting relations}
\label{sec:strategies_for_traditional_rewriting_relations}

In this section, we relate parallel rewriting to the traditional approach
to rewriting theory over vector spaces that consists in reducing one
basis element at each step. In particular, we show that the
$S$-confluence property implies confluence for traditional rewriting
relations, and we give a new proof of the diamond lemma, based on
$S$-confluence.
\medskip

As previously, we fix a set $R\subset\basis\times V$ of rewriting rules,
whose elements are written $e\rewR v$. Moreover, we also impose that for
every such rule, we have $e\notin\supp(v)$. These rules are extended into
a rewriting relation on $V$, still written $\rewR$, with rewriting steps
of the form
\begin{equation*}\label{equ:R_rewriting_step}
  \lambda e+u\rewR\lambda v+u,
  \smallskip
\end{equation*}
where $e\rewR v\in R$ is a rewriting rule, $\lambda$ is a nonzero scalar
and $u$ is a vector such that $e$ does not belong to $\supp(u)$. The
relation~$\rewR$ is not stable under vector spaces operations, since
$u_1\rewR u_2$, $v_1\rewR v_2$, and $\mu\in\K$ generally do not imply
$\mu u_1+v_1\rewR \mu u_2+v_2$. In contrast,
Proposition~\ref{prop:vs_structure} shows that $\equivR$ is compatible
with these operations. In the proof of this proposition, we use the
following lemma.
\smallskip

\begin{lemma}\label{lem:butterfly}
  If we have $u_1\equivR u_2$ and $v_1\equivR v_2$, then,
  $\mu u_1+v_1\equivR\mu u_2+v_2$ holds for every $\mu\in\K$.
\end{lemma}

\begin{proof}
  The proof is an adaptation
  of~\cite[Lemma 3.1.3]{GuiraudHoffbeckMalbos19}. We first show the
  following particular case:
  \begin{equation}\label{equ:sum_rew}
    u_1\rewR u_2\quad\Rightarrow\quad\mu u_1+v_1\equivR\mu u_2+v_1.
  \end{equation}
  By
  definition of $\rewR$, we have $u_1=\lambda e+u$ and $u_2=\lambda v+u$,
  where $e\rewR v$ is a rewriting rule,~$\lambda$ is a scalar and $e$
  does not belong to $\supp(u)$. Let $\nu$ be the coefficient of $e$ in
  $v_1$, so that we may write $v_1=\nu e+v_1'$, where $e$ does not belong
  to $\supp(v_1')$. Since $\mu u_1+v_1=(\mu\lambda+\nu)e+\mu u+v_1'$,
  $\mu u_2+v_1=\nu e+\mu\lambda v +\mu u+v_1'$, and, $e\notin\supp(v)$,
  we have $\mu u_1+v_1\transR(\mu\lambda+\nu)v+\mu u+v_1'
  \overset{*}{\leftarrow}\mu u_2+v_1$, which proves~\eqref{equ:sum_rew}.
  If $u_1\equivR u_2$, using~\eqref{equ:sum_rew}, an induction on the
  length of the path $u_1\equivR u_2$ shows that
  $\mu u_1+v_1\equivR\mu u_2+v_1$, and by an analogous argument, we have
  $\mu u_2+v_1\equivR\mu u_2+v_2$. Hence, we have
  $\mu u_1+v_1\equivR\mu u_2+v_2$, which concludes the proof.
\end{proof}
\smallskip

\begin{proposition}\label{prop:vs_structure}
  Given two vectors $u,u'\in V$, we have $u\equivR u'$ if and only if
  $u + \Span{R} = u'+\Span{R}$.
\end{proposition}

\begin{proof}
  By definition of the rewriting relation $\rewR$, if $u\rewR u'$, then
  $u + \Span{R} =u'+\Span{R}$. Since $\equivR$ is the smallest equivalence relation
  that contains $\rewR$, we deduce that $u\equivR u'$ also implies
  $u+\Span{R}=u'+\Span{R}$. Conversely, let us write $u-u'=\sum\mu(e-v)$, where
  $\mu$ are scalars and $e\rewR v$ are rewriting rules. Since
  $e\equivR v$ holds for each term of the sum, Lemma~\ref{lem:butterfly}
  implies that $u\equivR u'$, which concludes the proof.
\end{proof}
\smallskip

As a consequence of the previous proposition, $V/\equivR = V/\Span{R}$
admits a vector space structure, given by the operations
\[[u]_{\equivR}+\lambda[v]_{\equivR}=[u+\lambda v]_{\equivR}.\smallskip\]
\noindent
Notice that Proposition~\ref{prop:vs_structure} is not true if we do not
assume that $e\notin\supp(v)$ holds for every rewriting rule~$e\rewR v$: 
for instance, if we have only one rule $e\rewR 2e$, then $e\in\Span{R}$ 
but $[e]_{\equivR}\neq[0]_{\equivR}$. We point out that the assumption
$e\notin\supp(v)$ was explicitly used in the proof of
Lemma~\ref{lem:butterfly}.
\medskip

In Theorem~\ref{thm-S-conf}, we prove a confluence criterion for $\rewR$
in terms of strategies. Before that, we need the following lemma, that
relates parallel rewriting to traditional rewriting.
\smallskip

\begin{lemma}\label{lem:strategies}
  Let $S$ be a strategy for $R$. Then, the following inclusion holds:
  \[\parS\quad\subset\quad\transR.\]
\end{lemma}

\begin{proof}
  Since $S$ is a strategy, the preorder $\ordS$ of $S$ is a
  well-founded order. This order extends into a well-founded order on $V$
  defined by $u'\ordS u$ if $\supp(u')$ is smaller than $\supp(u)$ for
  the multiset order of $\ordS$. We show the proposition by induction
  \wrt\ $\ordS$. If $u$ is minimal, then $r_S(u)=u$ so that
  $u\transR r_S(u)$. Suppose now that $u$ is not minimal. Then $u$ may be
  uniquely written in the form 
  \begin{equation*}\label{equ:decompo_max}
    u=\sum_{i=1}^n\lambda_ie_i+u'
  \end{equation*}
  where the basis elements $e_i$ are the elements of $\supp(u)$ that are
  maximal for~$\ordS$, and the $\lambda_i$'s are their coefficients in
  $u$. In particular, we have $u'<_S u$ and by induction, we have
  $u'\transR r_S(u')$. By definition of $\ordS$, the rewriting rules that
  are involved in this rewriting sequence have left-hand sides strictly
  smaller than $e_i$'s, so that $u\transR\sum\lambda_ie_i+r_S(u')$.
  Moreover, since the $e_i$'s are not comparable for~$\ordS$, for each
  indices $i$ and $j$, $e_i$ does not belong to $\supp(r_S(e_j))$. Hence,
  we may reduce successively each $e_i$ into $r_S(e_i)$ and finally have
  \[
  u \quad\transR\quad
  \sum \lambda_ie_i+r_S(u')\quad\transR\quad
  \sum\lambda_ir_S(e_i) + r_S(u')
  =r_S(u).\]
\end{proof}
\smallskip

Now, we can show that the $S$-confluence property implies confluence of
$\rewR$.
\medskip

\begin{theorem}\label{thm-S-conf}
  Let $R$ be a set of rewriting rules such that for every
  $e\rewR v\in R$, we have $e\notin\supp(v)$, and let $S$ be a strategy
  for $R$. If $R$ is $S$-confluent, then the rewriting relation $\rewR$
  is confluent.
\end{theorem}

\begin{proof}
  It is sufficient to show that $\rewR$ has the Church-Rosser property,
  which can be done using our previous results. Let $u,u'\in V$ be two
  vectors such that $u\equivR u'$. From
  Proposition~\ref{prop:vs_structure}, we have $u+\Span{R}=u'+\Span{R}$,
  from Theorem~\ref{thm:SNF_decompo}, we have $\SNF(u)=\SNF(u')$,
  and from Lemma~\ref{lem:strategies} we have $u\transR\SNF(u)$ and
  $u'\transR\SNF(u')$. All together, we get that $u$ and $u'$ rewrite 
  into $\SNF(u)=\SNF(u')$, so that~$\rewR$ has the Church-Rosser
  property.
\end{proof}
\smallskip

Note that $S$-confluence is a sufficient but not a necessary condition for
confluence. Indeed, with~$\basis$ the set of integers and the rewriting rules
$n\rewR n+1$ as in Point~\ref{it:case_N} of
Example~\ref{ex:strategies_step_1}, there is no strategy such that $R$ is
confluent relative to this strategy.
\smallskip

\begin{example}\label{ex:conf}
  In Example~\ref{ex:S-conf}, we have shown that the following set of
  rewriting rules
  \[R=\{e_1\rewR e_2,\quad e_2\rewR e_3+e_4,\quad e_3\rewR e_2-e_4\}\]
  is $S$-confluent for the strategy defined in Point~\ref{it:ex_strat_1}
  of Example~\ref{ex:strategies_step_1}. Hence, $\rewR$ is confluent.
\end{example}
\smallskip

\begin{remark}
  
  Putting together Proposition \ref{prop:vs_structure} and Theorem \ref{thm-S-conf}, a sufficent condition for the confluence for $\rewR$ is the existence, for any rule $\rho = e \rewR v \in R$,
   of a diagram of the following shape:
   \[\begin{tikzcd}[sep = large]
    e\ar[d, "_R"']\ar[r, twoheadrightarrow, "_S"'] &
    r_S(e)  \\
    v \ar[ru, "_{R_{\prec \rho}}"', leftrightarrow, "*"] &.
  \end{tikzcd}\]
The two main features of these results -- using an order on rewriting steps and allowing  zig-zags in confluence diagrams -- are inspired from van Ostroom's work of decreasingness, introduced in \cite{van2008confluence} for abstract rewriting system.
  
\end{remark}

\smallskip

We finish this section by showing how the diamond lemma fits as a
particular case of our setup. For that, we first recall that the
rewriting preorder of $\rewR$ is the preorder on $\basis$ defined as
being the transitive closure of $e'\ordR e$ if there exists a vector
$u\in V$ such that there is a non empty rewriting sequence $e\transR u$
and $e'\in\supp(u)$. The rewriting preorder of $\rewR$ is still written
$\ordR$. The notation~$\ordR$ is the same than the one that we choose for
the rewriting preorder of a strategy since, even if they refer to
different objects, they are analogous: in fact $\ordS$ was defined as
being the adaptation of $\ordR$ to the case of parallel rewriting. Before
proving the diamond lemma using strategies, we need the following
preliminary result.
\smallskip

\begin{lemma}\label{lem:rew_preorder}
  Let $S$ be a prestrategy for $R$. If $e'\ordS e$, then we have
  $e'\ordR e$. In particular, if $\ordR$ is well-founded, then $S$ is a
  strategy for $R$.
\end{lemma}

\begin{proof}
  Let $e\in\basis$ be a basis element. We first show by induction that
  for every strictly positive integer~$n\geq 1$ and for every
  $e'\in\supp(r_S^n(e))$, we have $e'\ordR e$. Assuming that that holds
  for every strictly positive integer that is smaller or equal to $n-1$,
  we distinguish two cases. First, if $e'\in\supp(r_S^{n-1}(e))$ and
  $n\geq 2$, then we have $e'<_Re$ by induction hypothesis. In the other
  case ($n=1$ or $e'\notin\supp(r_S^{n-1}(e))$), then there exists
  $e''\in\supp(r_S^{n-1}(e))$, such that $e'\in\supp(r_S(e''))$, so that
  there is a rewriting rule $e''\rewR r_S(e')$ in $R$. That implies that
  $e'<_Re''$ and thus $e'<_Re$; the last assertion is due to the facts
  that if $n=1$, then $e''=e$ and if not, we have $e''<_Re$ by induction
  hypothesis. Since $\ordS$ is the transitive closure of there exists a
  strictly positive integer~$n$ such that $e'\in\supp(r_S^n(e))$ and
  since $\ordR$ is transitive, the first assertion is shown. This
  assertion implies that if $\ordR$ is well-founded, then $\ordS$ is also
  well-founded, hence the second assertion.
\end{proof}
\smallskip

\begin{theorem}[Diamond's Lemma \cite{MR506890}]
  \label{thm:diamond_lemma}
  Let $R$ be a set of rewriting rules such that the rewriting preorder
  $\ordR$ of~$\rewR$ is well-founded. Assume that  for every $e\in\basis$
  such that $e\rewR v$ and $e\rewR v'$, $v$ and $v'$ are joinable. 
  Then,~$\rewR$ is confluent.
\end{theorem}

\begin{proof}
  First, notice that since $\ordR$ is terminating, all the rules
  $e\rewR v\in R$ are such that $e\notin\supp(v)$. From
  Theorem~\ref{thm-S-conf}, we only have to prove that there exists a
  strategy $S$ for $R$ such that $R$ is $S$-confluent. We define $S$ as
  follows: for every basis element $e$ that is reducible by $\rewR$, we
  select exactly one arbitrary rewriting rule with left hand-side $e$ and
  we define $S$ as being the prestrategy composed of these selected
  rewriting rules. From Lemma~\ref{lem:rew_preorder}, $S$ is a strategy
  for $R$.
  \smallskip

  From Proposition~\ref{prop:decreasing}, in order to show that $R$ is
  $S$-confluent, we only have to construct a well-founded order $\prec$
  on $R$ such that $R$ is decreasing \wrt\ $(S,\prec)$. Given two
  rewriting rules $\rho=e\rewR v$ and $\rho'=e'\rewR v'$, we let
  $\rho\prec\rho'$ if and only if $e\ordR e'$. Let us consider an
  arbitrary rule $\rho=e\rewR v$. If $\rho$ belongs to $S$, then we have
  $v=r_S(e)$, hence $v-r_S(e)=0$ belongs to~$\Span{\Rrho}$. Otherwise,
  since $e$ is not a normal form for $R$, there exists a rule
  $e\rewR r_S(e)$ in $S$. Using the joinability hypothesis, we get that
  $v$ and $r_S(e)$ are joinable, so that we have a decomposition
  $v-r_S(e)=\sum\lambda_i(e_i-v_i)$, where $\rho_i=e_i\rewR v_i$ are part
  of the rules used to join $v$ to $r_S(e)$. By definitions of the orders
  $\ordR$ and $\prec$, we have $e_i\ordR e$, so that $\rho_i\prec\rho$.
  Hence, we have $v-r_S(e)\in\Span{\Rrho}$, which proves that $R$ is
  decreasing \wrt\ $(S,\prec)$. Hence, $R$ is $S$-confluent, so that
  $\rewR$ is confluent.
\end{proof}

\section{Rewriting strategies over rational Weyl algebras}
\label{sec:rewriting_strategies_over_rational_Weyl_algebras}

In this section, we investigate rewriting systems over rational Weyl
algebras and relate involutive divisions to rewriting strategies for such
systems. In particular, we show that involutive sets in rational Weyl
algebras induce confluent rewriting systems.
\medskip

Throughout the section, we fix a set $X=\{x_1,\cdots,x_n\}$ of
indeterminates and the field of fractions of the polynomial algebra $\QX$
over~$\Q$ is denoted by $\Q(X)=\QXX$, it is the set of rational
functions in the indeterminates $X$. We fix another set of
variables $\Delta=\{\diff{1},\cdots,\diff{n}\}$ that model partial
derivative operators, see Example~\ref{ex:diff_operators_init}. We denote
by $\partial^{\alpha}=\diff{1}^{\alpha_1}\cdots\diff{n}^{\alpha_n}$ the
monomial over $\Delta$ with multi-exponent
$\alpha=(\alpha_1,\cdots,\alpha_n)\in\N^n$. Finally, let $\monBasis$ be
the set of monomials over~$\Delta$:
\[\monBasis=\left\{\partial^\alpha:\ \alpha\in\N^n\right\}.
\smallskip\]
In what follows, we keep the terminology monomials only for elements of
$\Mon(\Delta)$ and not for elements in $\Mon(X)$.

\subsection{Rewriting systems over rational Weyl algebras}
\label{sec:rewriting_systems_over_Weyl_algebras}

In this section, we recall the definition of the rational Weyl algebra
and introduce rewriting systems induced by monic operators.
\medskip

\begin{definition}
  The {\it rational Weyl algebra} over $\Q(X)$ is the set of polynomials
  $\Q(X)[\Delta]$ with coefficients in $\Q(X)$ and indeterminates
  $\Delta$. The multiplication of this $\mathbb Q$-algebra is induced by
  the commutation laws $\partial_i\partial_j=\partial_j\partial_i$ and
  \[\diff{i}f=f\diff{i}+\frac{d}{dx_i}(f),\quad f\in\Q(X),\quad
  1\leq i\leq n,
  \smallskip\]
  where $d/dx_i:\Q(X)\to\Q(X)$ is the partial derivative operator with
  respect to~$x_i$. This algebra is denoted by $\Weyl{n}$.
\end{definition}
\smallskip

Notice that $\Weyl{n}$ is a $\Q(X)$-vector space and that the monomial 
set $\monBasis$ is a basis of~$\Weyl{n}$. Elements of~$\Weyl{n}$ should 
be thought of as differential operators whose coefficients are rational
functions, and for this reason, a generic element of this algebra is
denoted by $\D$ and is called a differential operator. In the following
example, we illustrate how these operators provide an algebraic model of
linear systems of ordinary differential (in the case $n=1$) and partial
derivative equations (in the case~$n\geq 2$) with one unknown function. 
\smallskip

\begin{example}\label{ex:diff_operators_init}
  {\color{white}toto}
  \begin{enumerate}
  \item\label{it:ODE_init} The linear ordinary differential equation
    $y'(x)=xy(x)$ is written in the form $(\D y)(x)=0$, where the operator
    $\D=\partial-x$ belongs to $\Weyl{1}=\Q(x)[\partial]$. 
  \item\label{it:Janet_example_init} Consider Janet's
    example~\cite{MR1308976}, that is, the linear system of partial
    derivative equations with~$3$ variables, one unknown function, and
    the two equations $y_{33}(x)=x_2y_{11}(x)$ and $y_{22}(x)=0$, where
    $y_{ij}(x)$ denotes the second order derivative of the unknown
    function $y(x)$ with respect to the variables $x_i$ and~$x_j$. Then,
    these equations are written $(\D_1y)(x)=0$ and $(\D_2y)(x)=0$, where
    $\D_1,\D_2\in\Weyl{3}$ are defined as follows:
    \[\D_1=\partial_3^2-x_2\partial_1^2,\quad \D_2=\partial_2^2.
    \smallskip\]
  \end{enumerate}
\end{example}

\begin{remark}
  In~\ref{it:ODE_init} of Example~\ref{ex:diff_operators_init}, we
  implicitly used that every $f\in\Q(X)$ induces a unique multiplication
  operator $y(x)\mapsto f(x)y(x)$.
\end{remark}
\smallskip

The next step before introducing rewriting systems over rational Weyl
algebras is to recall the definition of monic operators. We fix a
monomial order $\prec$ on $\monBasis$, that is, a well-founded total 
order which is admissible, \ie, $\partial^{\alpha}\prec\partial^{\beta}$
implies $\partial^{\alpha+\gamma}\prec\partial^{\beta+\gamma}$, for every
$\alpha,\beta,\gamma\in\N^n$. Given an operator $\D$, we denote by
$\lm(\D)$ the leading monomial of $\D$ with respect to~$\prec$, that 
is, $\lm(\D)$ is the greatest element of $\supp(\D)$, where the support 
is defined \wrt\ the basis $\monBasis$. 
\smallskip

\begin{definition}
  Let $\prec$ be a monomial order on $\monBasis$. A differential
  operator $\D\in\Weyl{n}$ is said to be $\prec$-{\em monic} if the
  coefficient of $\lm(\D)$ on $\D$ is equal to $1$. Moreover, given a
  monic differential operator $\D$, we denote by $r(\D)=\lm(\D)-\D$.
\end{definition}
\smallskip

Since the monomial order $\prec$ is fixed, me simply say monic instead of
$\prec$-monic. Given a set of monic operators~$\Theta\subseteq\Weyl{n}$,
let us consider the rewriting relation on $\Weyl{n}$ induced by the
following rewriting rules: 
\begin{equation}\label{equ:rewTheta}
  \RTheta=\Big\{\partial^\alpha\lm(\D)\to_{\RTheta}\partial^\alpha
  r(\D):\ \D\in\Theta,\ \partial^\alpha\in\Mon(\Delta)\Big\}.
\end{equation}
For simplicity, we write $\D\rewTheta\D'$ instead of
$\D\to_{R_\Theta}\D'$. The rewriting relation $\rewTheta$ is terminating
since the rewriting rules reduce a monomial into a combination of
strictly smaller monomials \wrt\ the well-founded order $\prec$.
Moreover, notice that in the case where the coefficient $\lc(\D)\in\Q(X)$
of~$\lm(\D)$ in $\D$ is not constant, the situation is much harder.
Indeed, in this case, the left-hand sides of the rewriting rules are of
the form $\partial^\alpha(\lc(\D)\lm(\D))$ and due to commutation laws,
these elements are not monomials. In particular, we are not in the
situation of our general approach developed in  
Section~\ref{sec:rewriting_strategies_over_vector_space} anymore.
\medskip

We finish this section with some comments on $\rewTheta$. Let us consider 
the linear system of ordinary differential or partial derivative
equations with unknown function $y$ given by 
\begin{equation}\label{equ:PDE_system}
  \{(\D y)=0:\D\in\Theta\}.
\end{equation}
Let $y(x)$ be an arbitrary solution to this system. Then, for every
operator $\partial^\alpha$ and every $\D\in\Theta$, we also have
$(\partial^\alpha\D y)(x)=0$, or equivalently,
$(\partial^\alpha\lm(\D)y)(x)=(\partial^\alpha r(\D)y)(x)$. Hence, if
there is a rewriting sequence $\D_1\transTheta\D_2$, then the solution
$y(x)$ of~\eqref{equ:PDE_system} satisfies $(\D_1y)(x)=(\D_2y)(x)$. This
remark has deep applications in the formal theory of partial differential
equations, for instance for finding integrability conditions or computing
dimensions of solution spaces, see~\cite{MR1308976}. Moreover, notice
that since $\monBasis$ is a commutative set, there is another possible
choice for rewriting the monomial $\partial^\alpha\lm(\D)$
in~\eqref{equ:rewTheta}. Indeed, we could swap $\partial^\alpha$ and
$\lm(\D)$ to get the new rule
$\lm(\D)\partial^\alpha\rewTheta r(\D)\partial^\alpha$. This rule is
simpler in the sense that it does not require to apply any commutation 
law to its right-hand side in contrast with~\eqref{equ:rewTheta}.
However, we do not take this rule into account since it would break the
algebraic model of partial derivative equations. Indeed, if $y(x)$ is a
solution of~\eqref{equ:PDE_system}, then the relation
$(\lm(\D_i)\partial^\alpha y)(x)=(r(\D_i)\partial^\alpha y)(x)$ does not
hold in general, as illustrated in~\ref{it:ODE_rew} of the following
example.
\smallskip

\begin{example}\label{ex:diff_operators_rew}
  We continue Example~\ref{ex:diff_operators_init}.
  \begin{enumerate}
  \item\label{it:ODE_rew} Let $\Theta=\{\D\}$ where
    $\D=\partial-x\in\Weyl{1}$. Since $\partial$ is greater than $1$ for
    every monomial order,~$\rewTheta$ is induced by the rewriting rules
    $\partial^n\rewTheta \partial^{n-1}x$, where $n$ is a strictly
    positive integer. In particular, we have the following rewriting
    sequence:
    \[\partial^2\rewTheta\partial x=x\partial+1\rewTheta x^2+1.
    \smallskip\]
    In terms of the corresponding differential equation $y'(x)=xy(x)$,
    this rewriting sequence has the following meaning. First, notice that
    the space of solutions of this equation is the one-dimensional
    $\R$-vector space spanned by the function $e^{x^2/2}$. Moreover, the
    second order derivative of a solution $y(x)=Ce^{x^2/2}$, for an
    arbitrary constant $C$, is given by the formula
    $y''(x)=(x^2+1)Ce^{x^2/2}$, 
    which reads $(\partial^2y)(x)=(x^2+1)y(x)$ in terms of operators.
    Notice that if we allow to reduce the left~$\partial$ 
    in~$\partial^2$, then we get $\partial^2\transTheta x^2$, which is 
    false in terms of operators since $y''(x)$ is not equal to~$x^2y(x)$.
  \item\label{it:Janet_example_rew} Let $\Theta=\{\D_1,\D_2\}$, where
    $\D_1=\partial_3^2-x_2\partial_1^2$ and $\D_2=\partial_2^2$
    correspond to the two equations of the Janet example. We define
    $\prec$ as being the deg-lex order on
    $\Mon(\partial_1,\partial_2,\partial_3)$ induced by
    $\partial_1\prec\partial_2\prec\partial_3$, so that $\rewTheta$ is
    induced by the rewriting rules
    $\partial_3^2\rewTheta x_2\partial_1^2$ and
    $\partial_2^2\rewTheta 0$. Then,~$\rewTheta$ is not confluent since:
    \begin{equation}\label{equ:non_conf_Janet_ex}
      \begin{tikzcd}
        \partial_2^2\partial_3^2\ar[d, "_\Theta"']\ar[r, "_\Theta"'] &
        \partial_2^2(x_2\partial_1^2)\ar[d, "_\Theta"]\\
        0 & 2\partial_1^2\partial_2
      \end{tikzcd}
    \end{equation}
    The right arrow is an application of the rule
    $\partial_2^2 \rewTheta 0$, made possible since
    $\partial_2^2(x_2\partial_1^2)$ is equal to
    $\partial_1^2\partial_2+x_2\partial_1^2\partial_2^2$ (to see this, it
    suffices to apply twice the commutation law
    $\partial_2x_2=x_2\partial_2+1$).  We deduce
    from~\eqref{equ:non_conf_Janet_ex} that any solution $y(x)$ of the
    equations $(\D_iy)(x)=0$ has to verify the new integrability
    condition $y_{112}(x)=0$.
  \end{enumerate}
\end{example}
\smallskip

\subsection{Involutive divisions and strategies}
\label{sec:involutive_divisions_and_strategies}

In this section, we interpret involutive divisions in terms of strategies
for the rewriting relation induced by a set of monic differential
operators. From this, we show that the rewriting system induced by an
involutive set of operators is confluent.
\medskip

We first recall from~\cite{MR1627129} the definition of involutive
divisions and associated notions that are involutive divisors,
multiplicative variables, and autoreducibility. For that, we temporally
work with monomials instead of operators and denote these monomials with
Latin letters $u,m$ instead of~$\partial^\alpha$. Then, we will reuse the
operator notation for monomials when we will consider rewriting systems
over rational Weyl algebras. An {\em involutive division} $L$ on
$\Mon(\Delta)$ is defined by a binary relation~$\divInv{L}^U$ on
$U\times\Mon(\Delta)$, for every finite subset $U\subset\Mon(\Delta)$,
satisfying for every $u,u'\in U$ and every $m,m'\in\Mon(\Delta)$, the
following relations:
\begin{enumerate}[label=\alph*)]
\item\label{it:div} $u\divInv{L}^Um\Rightarrow u\mid m$,
\item\label{it:unit} $u\divInv{L}^Uu$,
\item\label{it:mul} $u\divInv{L}^Uum$ and $u\divInv{L}^Uum'$ if and only
  if $u\divInv{L}^Uumm'$,
\item\label{it:vertex} $u\divInv{L}^Um$ and $u'\divInv{L}^Um$ implies
  $u\divInv{L}^Uu'$ or $u'\divInv{L}^Uu$,
\item\label{it:transitivity} $u\divInv{L}^Uu'$ and $u'\divInv{L}^Um$
  implies $u\divInv{L}^Um$,
\item\label{it:filter} for every $V\subseteq U$ and every $v\in V$,
  $v\divInv{L}^Um$ implies $v\divInv{L}^Vm$. 
\end{enumerate}
In the sequel, we write $\divInv{L}$ instead if $\divInv{L}^U$ when the
context is clear. We say that $u\in U$ is an {\em L-involutive divisor}
of a monomial $m$ if $u\divInv{L}m$. The variable~$\partial_i$ is said to be
{\em L-multiplicative} for $u$ \wrt\ $U$ if $u$ is an $L$-involutive
divisor of $\partial_iu$. Notice that $u\divInv{L}m$ if and only if
$m=m'u$, where $m'$ contains only $L$-multiplicative variables for $u$
\wrt\ $U$. Notice also that an involutive division is entirely
determined by the list of multiplicative variables \wrt\ each finite
set $U$ such that conditions \ref{it:vertex}, \ref{it:transitivity}, and
\ref{it:filter} are fulfilled. We say that $U$ is {\em L-autoreduced} if
every $u\in U$ admits only $u$ as $L$-involutive divisor, \ie,
$u'\divInv{L}u$ implies $u'=u$. Notice that if $U$ is $L$-autoreduced,
then every monomial $m$ admits at most one $L$-involutive divisor. We
finish this discussion on involutive divisions with three classical
examples. Before, let us introduce the following notation: given a
monomial $m=\partial^\alpha\in\monBasis$, let us denote by
$d_k(m)=\alpha_k$ the degree of $m$ \wrt the variable $\partial_k$.
\smallskip

\begin{example}\label{ex:involutive_division}

  We fix a finite set of monomials $U\subset\monBasis$. The
  {\em Janet, Thomas} and {\em Pommaret} divisions are the involutive
  divisions $\divInv{J},\divInv{T}$, and $\divInv{P}$ such that the
  variable $\partial_i$, where $1\leq i\leq n$, is
  $J,L$ or $P$-multiplicative for $u$ \wrt\ $U$ if and only if 
  \begin{itemize}
  \item for $\divInv{J}$: $d_i(u)=\max\{d_i(u'):\ u'\in U\ \text{and}\
    d_j(u')=d_j(u),\ \forall i<j\leq n\}$, 
  \item for $\divInv{T}$: $d_i(u)=\max\{d_i(u'):\ u'\in U\}$,
  \item for $\divInv{P}$: for every $1\leq j\leq i$, we have $d_j(u)=0$.
  \end{itemize}

\end{example}
\smallskip

Now, we return to differential operators and we fix a monomial order
$\prec$ on $\monBasis$. Given a finite set $\Theta\subset\Weyl{n}$ of
$\prec$-monic differential operators, all the theory of monomial sets can
be applied to the case where $U$ is the set of leading monomials of
elements of $\Theta$:
\[\lm(\Theta)=\left\{\lm(\D):\ \D\in\Theta\right\}\subset\monBasis
\smallskip\]
Hence, we may extend the autoreducibility property for monomial sets
\wrt\ an involutive division to sets of differential operators.
\smallskip

\begin{definition}
  Let $\Theta\subset\Weyl{n}$ be a finite set of $\prec$-monic
  differential operators, let $\prec$ be a monomial order, and let $L$ be
  an involutive division on $\Mon(\Delta)$. We say that $\Theta$ is
  {\em left L-autoreduced} if $\lm(\Theta)$ is $L$-autoreduced.
\end{definition}
\smallskip
\noindent
The adjective "left" is here to emphasis that it may exist
$\D,\D'\in\Theta$ such that $\lm(\D)$ is an $L$-involutive divisor of a
monomial $\partial^\alpha\in\supp(r(\D'))$.
\medskip

\begin{example}\label{ex:multiplicative_variables}
  We can now apply the involutive divisions of 
  Example~\ref{ex:involutive_division} to find the multiplicative
  variables associated to the differential operators of
  Example~\ref{ex:diff_operators_rew}.
  \begin{enumerate}
  \item Take $\Theta = \{\D\}$, where $\D = \partial - x \in \Weyl 1$.
    Then, $\lm(\D) =\partial$, and $\partial$ is a multiplicative
    variable for $\D$ for the Janet and Thomas divisions, but not for the
    Pommaret one. This means that $\partial \divInv{J}^\Theta \partial^n$
    and $\partial \divInv{T}^\Theta \partial^n$ for all $n > 0$, but that
    $\partial \nmid_P^\Theta \partial^n$, unless $n = 1$. In addition,
    since $\Theta$ is a singleton, it is trivially left-autoreduced for
    all three involutive divisions. 
  \item Take now $\Theta = \{\D_1,\D_2\}$, where
    $\D_1=\partial_3^2 - x_2\partial_1^2$ and $\D_2 = \partial_2^2$. The
    following table gives the multiplicative variables for $\D_1$ and
    $\D_2$ \wrt\ $\Theta$ for all three involutive divisions:
    \begin{center}
    \begin{tabular}{l|ccc}
      & Janet & Thomas & Pommaret \\ \hline
      $\D_1$ & $\partial_1, \partial_2, \partial_3$ & $\partial_1, \partial_3$ & $\emptyset$ \\
      $\D_2$ & $\partial_1, \partial_2$ & $\partial_2$ & $\partial_1$ \\
    \end{tabular}
  \end{center}
    Once again, the leading monomials of elements of $\Theta$ do not divide
    each others, so $\Theta$ is left-autoreduced for all three involutive
    divisions.
  \end{enumerate}
\end{example}
\smallskip

From now on, we fix a set $\Theta$ of monic (the order being fixed, we
drop it in $\prec$-monic) differential operators. Let $\RTheta$  be the
set of rewriting rules of the form
$\partial^\alpha\lm(\D)\rewTheta\partial^\alpha r(\D)$, such as
in~\eqref{equ:rewTheta}. Since~$\lm(\Theta)$ is the only monomial set we
will work with, we omit it in the symbol of the involutive division: we
write $\lm(\D)\divInv{L}\partial^\alpha\lm(\D)$ when $\partial^\alpha$
contains only $L$-multiplicative variables for $\lm(\D)$ \wrt\
$\lm(\Theta)$. Finally, we let
\begin{equation}\label{equ:S-strategy}
  \SThetaL=\Big\{\partial^\alpha\lm(\D)\parTheta{L}\partial^\alpha
  r(\D) : \, \D\in\Theta,\quad\lm(\D)\divInv{L}\partial^\alpha\lm(\D)
  \Big\}.
  \smallskip
\end{equation}
Here again, we choose to write
$\partial^\alpha\lm(\D)\parTheta{L}\partial^\alpha r(\D)$ instead of
$\partial^\alpha\lm(\D)\twoheadrightarrow_{\SThetaL}\partial^\alpha
r(\D_i)$ in order to simplify notations.
\smallskip

\begin{proposition}\label{prop:involutive_strategy}
  Let L be an involutive division on $\Mon(\Delta)$ such that $\Theta$ is
  left L-autoreduced. Then $\SThetaL$ is a strategy for $\RTheta$.  
\end{proposition}

\begin{proof}
  If the set $\Theta$ is left $L$-autoreduced, then every monomial admits
  at most one $L$-involutive divisor. Moreover, every left-hand side
  $\partial^\alpha\lm(\D)$ of a rewriting rule of $\SThetaL$ is
  $L$-divisible by~$\lm(\D)$. Hence, left-hand sides of $\SThetaL$ are
  pairwise distinct, which means that $\SThetaL$ is a pre-strategy for
  $\RTheta$. Finally, if $<_\Theta$ denotes the rewriting preorder of
  $\rewTheta$, then $\partial^\alpha<_\Theta\partial^\beta$ implies that
  $\partial^\alpha\prec\partial^\beta$, so that $<_\Theta$ is
  well-founded. From Lemma~\ref{lem:rew_preorder}, we get that $\SThetaL$
  is a strategy for~$\RTheta$. 
\end{proof}
\smallskip

From Proposition~\ref{prop:involutive_strategy}, any involutive division
$L$ such that $\Theta$ is left $L$-autoreduced induces a 
strategy~$\SThetaL$ for $\RTheta$. Hence, we get a well-defined
normalisation operator $\SThetaLNF$ corresponding to this strategy. The
following definition is an adaptation of the notion of involutive bases
for polynomial ideals~\cite{MR1627129} to the case of sets of monic
differential operators.
\smallskip

\begin{definition}
  Let $\Theta\subset\Weyl{n}$ be a finite set of differential operators,
  let $\prec$ be a monomial order on $\monBasis$ such that each element
  of $\Theta$ is monic, and let $L$ be an involutive division on
  $\Mon(\Delta)$ such that $\Theta$ is left $L$-autoreduced. We say that
  $\Theta$ is an {\em $L$-involutive set} if for every $\D\in\Theta$ and every
  $\partial^\alpha\in\Mon(\Delta)$, we have
  $\SThetaLNF(\partial^\alpha\D)=0$. 
\end{definition}
\smallskip

\begin{example}
Let us continue Example~\ref{ex:multiplicative_variables}.
\begin{enumerate}
\item In the case $\Theta = \{ \D \}$, with $\D = \partial - x$. For the
  Pommaret division, we have seen that $\D$ admits no multiplicative
  variable, so the strategy $S_{\Theta,P}$ is reduced to the rule
  $\partial\parTheta{P} x$. As a result we get:
  \[\partial \D = \partial^2 - \partial x = \partial^2 - x \partial - 1
  \parTheta{P} \partial^2 - x^2 - 1.
  \]
  This last term is a normal form for $\parTheta{P}$, hence
  $\SThetaNF{P}(\partial\D) \neq 0$ and so $\Theta$ is not
  $P$-involutive. On the other hand for the Janet and Thomas divisions,
  $S_{\Theta,J}$ and $S_{\Theta,T}$ coincide, and contain the rules
  $\partial^{n+1}\parTheta{L}\partial^n x$, where $L=J,T$. This yields:
  \[
  \partial\D=\partial^2-\partial x=\partial^2-x\partial-1\parTheta{L}
  \partial x-x^2-1=x\partial -x^2 \parTheta{L}0.
  \]
  So we get $\SThetaNF{L}(\partial\D) = 0$, and more generally
  $\SThetaNF{L}(\partial^n\D) = 0$: $\Theta$ is both  $J$- and
  $T$-involutive.  
\item In the case $\Theta = \{ \D_1 , \D_2 \}$, with
  $\D_1 = \partial_3^2 - x_2 \partial_1^2$ and $\D_2 = \partial_2^2$,
  $\Theta$ will not be involutive for either of the three involutive
  divisions of Example~\ref{ex:involutive_division}. In the case of the
  Janet division for example, we have:
  \[
    \partial_3^2 \D_2 = \partial_2^2 \partial_3^2
    \parTheta{J} \partial_2^2 (x_2 \partial_1^2) =
    x_2 \partial_1^2 \partial_2^2 - 2 \partial_1^2 \partial_2
    \parTheta{J} 2 \partial_1^2 \partial_2.
  \]
  This last term is a normal form for $S_{\Theta,J}$, so we get
  $\SThetaNF{J}(\partial_3^2 \D_2) = 2 \partial_1^2 \partial_2
  \neq 0$: $\Theta$ is not $J$-involutive. 
\end{enumerate}
\end{example}

The astute reader may remark that the last computation of the previous
example is closely related to the diagram appearing in 
Example~\ref{ex:diff_operators_rew}, which shows that $\rewTheta$ fails
to be confluent. This relationship between confluence and
$L$-involutivity is actually a very general one, as shown by the
following theorem.
\medskip

\begin{theorem}\label{thm:involutive_conf}
  Let $\Theta\subset\Weyl{n}$ be a finite set of differential operators,
  let $\prec$ be a monomial order on~$\monBasis$ such that each element
  of $\Theta$ is monic, and let L be an involutive division 
  on~$\monBasis$ such that $\Theta$ is left L-autoreduced. If $\Theta$ is
  L-involutive, then the rewriting relation~$\rewTheta$ is confluent.
\end{theorem}

\begin{proof}
  Let $\SThetaL$ be the strategy for $\RTheta$ defined such as
  in~\eqref{equ:S-strategy}. Since rewriting rules of $\RTheta$ are of
  the form $\partial^\alpha\lm(\D)\rewTheta\partial^\alpha R(\D)$, where
  $\D\in\Theta$ and $\partial^\alpha\in\monBasis$, the assumption that
  $\Theta$ is $L$-involutive means that $\rewTheta$ is
  $\SThetaL$-confluent. By Theorem~\ref{thm-S-conf}, $\rewTheta$ is
  confluent.
\end{proof}
\smallskip

\begin{remark}
  As for term rewriting systems or \G\ bases theory, there exists a
  completion procedure in the situation of differential operators, which
  corresponds to Knuth-Bendix or Buchberger procedures. In the case of 
  the Janet example, it turns out that after a finite number of steps,
  this procedure yields the the following involutive set,
  see~\cite{MR1308976}:
  \[\overline{\Theta}=\left\{\D_1,\quad\D_2,\quad\partial_1^2\partial_2,
  \quad\partial_2^2\partial_3,\quad\partial_1^4,\quad\partial_1^2
  \partial_2\partial_3,\quad\partial_1^4\partial_3\right\}.\]
\end{remark}
\medskip

The end of this section aims to show that
axioms~\ref{it:div}--\ref{it:transitivity} in the definition of an
involutive division may be formulated in a purely rewriting language
using strategies. We fix a strategy~$S$ for $\rewTheta$. For every
$\D\in\Theta$, we say that $\lm(\D)$ {\em S-divides} the monomial
$\partial^\alpha\in\monBasis$ if $S$ contains a rewriting rule of the
form $\partial^\alpha\lm(\D)\parS\partial^\alpha r(\D)$ and we say that
the variable $\partial_i\in\Delta$ is {\em S-multiplicative} for $\D$ if
$\partial_i\lm(\D)$ is $S$-divisible by $\lm(\D)$. 
\smallskip

\begin{definition}
  A strategy $S$ for $\RTheta$ is said to be {\em involutive} if for
  every left-hand side $\partial^\alpha\lm(\D)$ of a rewriting rule in
  $S$, then $\partial^\alpha$ contains only $S$-multiplicative variables
  of $\D$.
\end{definition}
\smallskip

\begin{proposition}
  If the strategy $S$ is involutive, then the S-division satisfies
  axioms~\ref{it:div}--\ref{it:transitivity} of the definition of an
  involutive division. Moreover, if L is an involutive division on
  $\monBasis$ such that $\Theta$ is left L-autoreduced, then the
  $\SThetaL$-division is the restriction of L to $\lm(\Theta)$.
\end{proposition}

\begin{proof}
  Let us show the first assertion. Axioms~\ref{it:div}, \ref{it:vertex},
  and~\ref{it:transitivity} hold since $S$ is a strategy for $\RTheta$.
  Indeed, left-hand sides of $\RTheta$ are of the form
  $\partial^\alpha\lm(\D)$, hence~\ref{it:div}, and left-hand sides of
  elements of $S$ are pairwise distinct, hence~\ref{it:vertex}
  and~\ref{it:transitivity}. Moreover, axioms~\ref{it:unit}
  and~\ref{it:mul} hold by definition of an involutive strategy.

  Let us show the second assertion. By definition of the strategy
  $\SThetaL$ and of the $\SThetaL$-division, $\lm(\D)$ has the same set
  of multiplicative variables for $L$ and for the $\SThetaL$-division.
  Hence, a monomial~$\partial^\alpha$ is $L$-divisible by $\lm(\D)$, with
  $\D\in\Theta$, if and only if it is $\SThetaL$-divisible by $\lm(\D)$.
  That proves the assertion.
\end{proof}

\section{Conclusion and perspectives}

In this paper, we considered rewriting systems over vector spaces, where
we proposed an alternative approach to the traditional one, since we used
parallel rewriting steps. We also established some links with the
traditional approach, by giving a confluence criterion as well as a proof
of the diamond lemma, based on strategies. Finally, we showed that our
general framework may be adapted to rational Weyl algebras, where
coefficients do not commute with monomials. In particular, we proved that
an involutive set in a rational Weyl algebra induces a confluent
rewriting system on it. We now present some possible extensions of our
work.
\medskip

A first research direction is to investigate the so-called
{\em standardisation properties}~\cite{Mellies05jwklop} associated to a
rewriting strategy. Indeed, the choice of strategy is nothing but the
choice for every vector of a preferred rewriting sequence starting at 
this vector. Moreover, we have shown in Lemma~\ref{lem:strategies} that
each elementary rewriting step for a strategy has same source and target
points than a rewriting sequence involving rules that do not belong to
the strategy. The proof of this fact is based on complete developments of
residuals, that play a central role in standardisation results,
corresponding to left-hand sides of the strategy. As a particular case,
we hope to interpret Janet bases in terms of standardisation.
\medskip

A second research direction is to extend our work to other algebraic
structures than vector spaces. This looks promising since we do away with
the notion of well-formed rewriting step, specific to the vector space
case. More generally, we hope to be able to extend our results to an
arbitrary category $\mathcal C$ (satisfying some suitable condition), 
recovering abstract rewriting in the case where $\mathcal C$ is the
category of sets, and linear rewriting as presented in this work in the
case where $\mathcal C$ is the category of vector spaces. Instead of
having a set or a vector space of terms to be rewritten, one would then
have an object of terms, which would be an object of $\mathcal C$. 
\medskip

A last research direction consists in applying rewriting systems over
rational Weyl algebras to the formal analysis of systems of partial
differential equations. As mentioned above, this topic covers many kinds
of problems and many techniques coming from rewriting theory and algebra
may be used in this context. We may focus on using rewriting methods
applied to the {\em Spencer cohomology}~\cite{MR1308976}, which, roughly
speaking, provides intrinsic properties, namely {\em $2$-acyclicity} and
{\em formal integrability}, that guarantee existence of normal form power
series solutions. 

\bibliography{Source}

\end{document}